\documentclass{llncs}

\usepackage[T1]{fontenc}
\usepackage{color}
\usepackage{graphicx}
\usepackage{amsfonts,enumerate}
\usepackage{latexsym,amssymb,amsmath}
\usepackage[ruled]{algorithm}
\usepackage[noend]{algpseudocode}
\algdef{SE}[DOWHILE]{Do}{doWhile}{\algorithmicdo}[1]{\algorithmicwhile\ #1}

\algdef{SE}[SUBALG]{Indent}{EndIndent}{}{\algorithmicend\ }%
\algtext*{Indent}
\algtext*{EndIndent}

\def\max{{\rm max}}
\def\min{{\rm min}}
\newtheorem{observation}{Observation}

\def\max{{\rm max}}
\def\patrol{\textsc{PUF}}

\newcommand{\reals}{\mathbb{R}}

\newcommand{\ignore}[1]{}

\begin{document}


\title{Patrolling a Path Connecting a Set of Points with Unbalanced Frequencies of Visits\thanks{This is the full version of the paper with the same title which
will appear in the proceedings of SOFSEM 2018,  44th International
Conference on Current Trends in Theory and Practice of Computer
Science, January 29 - February 2, 2018, Krems an der Donau, Austria.}}

\author{Huda Chuangpishit \inst{1} 
\and
Jurek Czyzowicz
\inst{2}\thanks{Research supported in part by NSERC.}
\and
Leszek~G\k asieniec \inst{3}\thanks{This work was in part supported by Networks Sciences and Technologies (NeST).} 
\and 
\\ Konstantinos Georgiou \inst{1}$^{\star\star}$ \and
Tomasz Jurdzi\'nski \inst{4}\thanks{Research supported by the Polish National
Science Centre grant DEC-2012/06/M/ST6/00459.} \and 
Evangelos Kranakis \inst{5}$^{\star\star}$}
\institute{{Department of Mathematics}, {Ryerson University}, {Toronto, Canada},
\email{hoda.chuang@gmail.com, konstantinos@ryerson.ca}
\and {D\'epartement d'informatique}, {Universit\'e du Qu\'ebec en Outaouais}, {Canada},
\email{Jurek.Czyzowicz@uqo.ca}
\and {Department of Computer Science}, {University of Liverpool}, {Liverpool, UK},
\email{L.A.Gasieniec@liverpool.ac.uk}
\and {Instytut Informatyki}, {Uniwersytet Wroc\l awski}, {Wroc\l aw, Poland},
\email{Tomasz.Jurdzinski@ii.uni.wroc.pl}
\and {School of Computer Science}, {Carleton University}, {Ottawa, Canada},
\email{evankranakis@gmail.com}
}

\maketitle

\begin{abstract}
Patrolling consists of scheduling perpetual movements of a collection of mobile robots, so that each point of the environment is regularly revisited by any robot in the collection. In previous research, it was assumed that all points of the environment needed to be revisited with the same minimal frequency.

In this paper we study efficient patrolling protocols for points located on a path, where each point may have a different constraint on frequency of visits. The problem of visiting such divergent points was recently posed  by G\k asieniec et al. in~\cite{GK+16}, where the authors study  protocols using a single robot patrolling a set of $n$ points located in nodes of a complete graph and in Euclidean spaces. 
  
The focus in this paper is on patrolling with two robots. We adopt a scenario in which all points to be patrolled are located on a line.  We provide several approximation algorithms concluding with the best
currently known $\sqrt 3$-approximation.

\end{abstract}

\section{Introduction}\label{sec:Intro}
In this paper we study efficient patrolling protocols by two robots for a collection of 
$n$ points distributed arbitrarily on a path or a segment of length~$1.$ 
Each point needs to be attended perpetually with {\sl known} but often distinct minimal 
frequency, i.e., some points need to be visited more often than others.

The problem was recently studied in~\cite{GK+16} 
where a collection of $n$ points was monitored with use of a single mobile robot. The points to be patrolled in~\cite{GK+16} are located in nodes
of a complete graph with edges of uniform (unit) length, as well as in Euclidean spaces,
where the points are distributed arbitrarily. In their work the frequency constraints refer
to {\em urgency factors} $h_i$, meaning that during a unit of time the urgency of point
$p_i$ grows by an additive term $h_i,$ and the task is to design a schedule of perpetual
visits to nodes which minimizes the maximum ever observed urgency on all points.   
In complete graphs and for any distribution of frequencies (urgency factors) the authors of \cite{GK+16} proposed a 2-approximation algorithm based on a reduction to the pinwheel scheduling problem, see, e.g., \cite{chan_general_1992,chan_schedulers_????,holte_pinwheel:_1989,holte_pinwheel_1992,lin_pinwheel_????}. 
They also discuss more tight approximations for the cases with more balanced urgency
factors. In Euclidean spaces \cite{GK+16} proposes several lower bounds and concludes with an
$O(\log n)$-approximation for an arbitrary distribution of points and urgency factors.

In our formulation, we assume that both robots have unit speed, and we try to minimize the relative violation of visitation-frequency requirements, i.e. the worst case time between two visitations over the required largest waiting time of each point.
Equivalently, one may think of the problem of finding the minimum possible speed $s$ that both robots should patrol with that induces no violation for the visitation-frequency requirements. 
In such setting, our {\em patrolling} result refers naturally to a competitive ratio, which is defined by the ratio of 
the speed the robots attain in our algorithm divided by the speed in the optimal solution. 

Specific to our model is the use of two robots, for which, as we show, one can achieve $\sqrt{3}$-approximation patrolling schedules. 
Notably, and maybe counter-intuitively, reducing the number of robots from two to one does not lead to constant approximation. 
An instructive example is when the central point has a very large visiting frequency (we can dedicate one robot to this point) 
comparing to the rest of the points on the line.  

In the previous research on boundary and fence patrolling (cf. \cite{czyzowicz_boundary_2011,CGK+15,KaKo15}) all points of the patrolled environment were supposed to be revisited with the same frequency.  However, assigning different importance to distinct portions of the monitored boundary seems natural and observable in practice. A particular variation of this problem was studied in \cite{collins_optimal_2013}, where the authors focus on 
monitoring {\em vital} (possibly disconnected) parts of a linear environment, while the remaining {\em neutral} portions of the boundary need not be attended at all. 

The problem of distinct attendance assigned to different portions of the environment, while of inherent combinatorial interest, is also observed in perpetual testing 
of virtual machines in cloud systems~\cite{AKG15}.
In such systems the frequency with which virtual machines are tested for undesirable 
symptoms may vary depending on the importance of dedicated cloud operational mechanisms.

The problem studied here is also a natural extension of several classical combinatorial
and algorithmic problems referring to {\em monitoring} and {\em mobility}.
This includes the {\em Art Gallery Problem}~\cite{ntafos_gallery_1986,orourke_1987} and its dynamic
variant called the {\em $k$-Watchmen Problem} \cite{urrutia_art_2000}.
In a more recent work on {\em fence patrolling} \cite{collins_optimal_2013,czyzowicz_boundary_2011,KaKo15} the efficiency of patrolling is measured by the {\em idleness} of the protocol, which is the  time where a point remains unvisited (maximized over all time moments and all points of the environment). 
In~\cite{CGK+15} one can find a study on monitoring of linear environments by robots prone to faults.
In~\cite{czyzowicz_boundary_2011,KaKo15} the authors assume robots have
distinct maximum speeds which makes the design of patrolling protocols more complex,
in which case the use of some robots becomes obsolete. 

In a very recent work~\cite{LS17} Liang and Shen consider a line of points attributed with
uniform urgency factors. For robots with uniform speeds, they give a polynomial-time
optimal solution, and for robots with constant number of speeds they present 
a 2-approximation algorithm. For an arbitrary number of velocities they design a 
4-approximation algorithm, which can be extended to 
a $2\alpha$-approximation algorithm family scheme, where integer $\alpha>1$ 
is the tradeoff factor to balance the time complexity and approximation ratio. 

\section{Problem Statement \& Definitions}
\label{sec: definitions}

An instance of the \textit{Path Patrolling Problem of Points with Unbalanced Frequencies (\patrol)}  consists of points $S=\{y_i\}_{i=1,\ldots,n}$ in the unit interval, where $0=y_1<y_2<\ldots<y_n=1$. Each point $y_i$ is associated with its \textit{idleness time} $I(y_i) \in \reals_+$, a positive real number which is also referred to as \textit{visitation frequency requirement}.

A perpetual movement schedule of two robots $r_1,r_2$ of speed 1 will be referred to as a \textit{patrolling schedule} (robots may change movement direction instantaneously, and at no cost). Given a patrolling schedule $\mathcal A$, we define the \textit{waiting time} $w_{\mathcal A}(y_i)$ of each point $y_i$ as the supremum of the time difference between any two subsequent visitations by any of $r_1,r_2$. When the patrolling schedule is clear from the context, we will drop the subscript in $w_{\mathcal A}$. 

A patrolling schedule $\mathcal A$ is called \textit{feasible} if for all $i$, $w_{\mathcal A}(y_i) \leq I(y_i)$. Schedule $\mathcal A$ is called \textit{$c$-feasible}, or \textit{$c$-approximation}, if $w_{\mathcal A}(y_i)/I(y_i)\leq c$, for each $i=1,\ldots,n$. Thus a feasible patrolling schedule is also $1$-approximation, or 1-feasible. 

An instance of \patrol\ that admits a feasible patrolling schedule will be called \textit{feasible}. In this paper we are concerned with the combinatorial optimization problem of minimizing the worst (normalized) violation of the idleness times for feasible instances, i.e., we are concerned with finding good approximation patrolling schedules, in which robots' trajectories can be determined efficiently in the size of the given input. We will call such patrolling schedules \textit{efficient}. 

The problem considered here is a close relative of {\em Pinwheel scheduling}~\cite{holte_pinwheel:_1989} 
modeled by points with non-uniform deadlines (visitation-frequencies) spanned by a complete network with edges of uniform length. 
The complexity of Pinwheel scheduling depends on its representation.
In particular we know that in the standard multi-set representation the problem is in NP, however, we still don't known whether it is NP-hard.
One can try to get closer to this answer either by studying particular instances of the problem which can be decided~\cite{holte_pinwheel_1992} 
or instead by seeking approximate solutions~\cite{GK+16}. In this paper we adopted the latter.    


We use the following concepts in the analysis of of our patrolling schedules.
We associate each point $y_i$ with its \textit{range} defined as the closed intervals
$ R(y_i) = \left[\max\left\{0, y_i-\frac{I(y_i)}{2}\right\}, \min\left\{1, y_i + \frac{I_(y_i)}{2}\right\}\right].$
Intuitively, $R(y_i)$ is the ball around $y_i$ within which a robot can be moving introducing no violation to the visitation frequency requirement of $y_i$. We also group points $y_i$ with respect to whether the extreme points fall within their range, i.e., we introduce:
\begin{align*}
S_{00}&:= \left\{ y_i \in S:~ 0,1\not \in R(y_i)\right\}, 
S_{01}:= \left\{ y_i \in S:~ 0\not \in R(y_i) \ni 1 \right\}, \\
S_{10}&:= \left\{ y_i \in S:~ 0 \in R(y_i) \not \ni 1 \right\}, 
S_{11}:= \left\{ y_i \in S:~ 0,1 \in R(y_i)\right\}.
\end{align*}

\section{Summary of Results \& Paper Organization}

Our main contribution pertains to efficient patrolling schedules (algorithms) of feasible \patrol\ instances. In particular, the patrolling schedules we propose are highly efficient and simple, meaning that robots' trajectories can be determined by a few critical turning points, which can be computed in linear time in the number of points of the \patrol\ instance. In order to do so, we provide in Section~\ref{sec: on feasible instances} some useful properties that all feasible \patrol\ instances exhibit, and in particular a characterization of instances with ``no problematic points''. For the latter instances, we also provide optimal feasible schedules (Theorem~\ref{thm:no-00-point}). 
Then we turn our attention to arbitrary feasible \patrol\ instances. As  a warm-up, we present in Section~\ref{sec: 4-feasible} a simple efficient $4$-approximation patrolling schedule that does not require coordination between robots. Section~\ref{sec: sqrt3-feasible} is devoted to the introduction of an elaborate and efficient $\sqrt{3}$-approximation patrolling schedule. The execution of the patrolling schedule requires robots to remember at most two special turning points (that can be found efficiently), and, in some cases, their coordination so that they never come closer than a predetermined critical distance. Its performance analysis is based on further properties of feasible \patrol\ instances that are presented in Section~\ref{sec:observations}. In particular, the $\sqrt{3}$-feasible patrolling schedule is the combination of Algorithms~\ref{alg:partition-patrolling} and~\ref{alg:sqrt(3)}, presented in Sections~\ref{sec: 1+2a-approx} and~\ref{sec: (2+a)/(1+a)-approx} respectively, each of them performing well for a different spectrum of a special structural parameter of the given instance that we call expansion. 
Finally in Appendices~\ref{sec: tight Algo1},~\ref{sec: tight Algo2} 
we also show that the analyses we provide for all our proposed algorithms are actually tight.

\section{Characterization of (Some) Feasible \patrol\ Instances}
\label{sec: on feasible instances}
In this section we characterize feasible instances of \patrol\ for which at least one of the extreme points falls within the range of each point. 
\begin{theorem}
\label{thm:no-00-point}
An instance of \patrol\ with $S_{00}=\emptyset$ is feasible  if and only if the following conditions are satisfied:
\begin{enumerate}[(1)]
\item $S_{10}\subset \bigcap_{x\in S_{10}} R(x)= X_{10}$, and $0\in X_{10}$.
\item $S_{01}\subset \bigcap_{x\in S_{01}} R(x)= X_{01}$, and $1\in X_{01}$.
\item $S\subset [\bigcap_{x\in S_{10}} R(x)] \cup \bigcap_{x\in S_{01}} R(x)] = X_{10}\cup X_{01}$
\end{enumerate}
Moreover, if conditions (1)-(3) are satisfied, then there exists an efficient 1-approximation partition-based patrolling schedule, i.e. a schedule in which every $y_i$ is visited only by one robot. 
\end{theorem}

In order to prove Theorem~\ref{thm:no-00-point} we need few observations.
\begin{observation}
\label{obs:nec1}
Assume $\mathcal{A}$ is a feasible patrolling schedule. Then, for each $x\in S$ and each time window of length at least $\frac{I(x)}{2}$ during an execution of $\mathcal{A}$, at least one robot is in $R(x)$.
\end{observation}
\begin{proof}
Reset time to $t_0=0$. Aiming at contradiction, assume there is no robot in $R(x)$ at $t\geq \frac{I(x)}{2}$. Since both robots have speed 1, no robot visited $x$ in the
period $[t-\frac{I(x)}{2},t]$ and no robot is able to visit $x$ in the period $[t, t+\frac{I(x)}{2}]$. Thus, $\mathcal{A}$ is not a feasible patrolling schedule.
\qed
\end{proof}
For simplicity, we may also assume that in any patrolling schedule (hence in feasible schedules as well), the position of robot $r_1$ in the unit interval is always to the left of the position of $r_2$, as otherwise we can exchange the roles of the robots whenever they swap while they meet. We summarize as follows. 
\begin{observation}
\label{obs:order}
In any patrolling schedule of \patrol, $r_1$ ($r_2$) is the only robot patrolling $y_1=0$ ($y_n=1$), and $r_1$ stays always to the left of $r_2$. 
\ignore{
Assume that there exists a feasible 2-robot solution. Then, there exists a feasible
solution in which, at each time $t$, the order of the position of the robots does not change during the execution of the solution.
}
\end{observation}

We are now ready to prove Theorem~\ref{thm:no-00-point}.

\begin{proof}[Theorem \ref{thm:no-00-point}]
First, we show implication $(\Rightarrow)$ by contraposition. If Condition (1) is not satisfied, then there exists $x\in S_{10}$ such that $x\notin X_{10}$.  Fix a feasible schedule $\mathcal A$. By Observation \ref{obs:order}, we may assume that $r_1$ stays to the left of $r_2$, throughout the execution of the schedule. By Observation \ref{obs:nec1}, there must be a robot in $X_{10}$ at each time $t$. Thus, $r_1$ must be in $X_{10}$ at each time $t$. Consequently, $x\in S_{10}\setminus X_{10}$ is visited only by $r_2$. But $r_2$ has to visit point $y_n=1$, and by definition of $S_{10}$ we know that $1\notin R(x)$. Therefore, $\mathcal A$ is not a feasible schedule.
By definition of $S_{10}$,  for all $x\in S_{10}$, we have $0\in R(x)$. Therefore $0\in X_{10}$.
A similar argument proves that Condition (2) is satisfied. 

By (1) and (2), there exist $a,b\in (0,1)$ such that $X_{10}=[0,a]$ and $X_{01}=[b,1]$. Now suppose that Condition (3) is not satisfied. Then $a<b$, and there is a point $x\in S$ such that $a<x<b$, and therefore neither $r_1$ nor $r_2$ can visit $x$.

For implication $(\Leftarrow)$, assume that (1)-(3) are satisfied. Consider a partition traversal $A$, where
$r_1$ is searching $X_{10}\setminus X_{01}$ and $r_2$ is searching $X_{01}$. Then, by the definition of the ranges $R(x)$, $X_{10}$ and $X_{01}$, the traversal $A$ is feasible.
\qed
\end{proof}

The complication of instances when $S_{00}$ is non empty is that in a feasible solution, points in $S_{00}$ have to be interchangeably patrolled by both $r_1,r_2$, which further requires appropriate synchronization between them. Even though a characterization of feasibility for such instances is eluding us, we provide below a necessary condition. This condition will be useful also later on. 
\ignore{
Consider a point $x\in S_{00}$. Then by definition of $S_{00}$ we know that $0,1 \notin R(x)$. On the other hand by Observation \ref{obs:nec1} a robot is always present in $R(x)$. But the robot that visits $x$ at time $t$ cannot afford to traverse to visit either 0 or 1 and return to $x$ for its next visit. So either a robot must be assigned to traverse $R(x)$ perpetually or the point $x$ must be visited by both robots $r_1$ and $r_2$. Therefore if a feasible solution exists, the robots may require synchronized trajectories to ensure one of them is traversing $R(x)$ while the other one is visiting the other points.

Next we consider the inputs for which $S_{00}\neq\emptyset$. In this case we present a necessary condition for the existence of a feasible solution. 
}
\begin{lemma}
\label{lem:nec-condition}
For every feasible instance of \patrol, we have  $S_{00}\subset\bigcap_{x\in S} R(x)$.
\end{lemma}
\begin{proof}
Suppose to the contrary, that there are $x\in S_{00}$ and $y\in S$, such that $x\notin R(y)$. By Observation \ref{obs:nec1}, a robot is always present inside $R(y)$. Therefore the other robot must visit $x$. Without loss of generality assume that $y<x$. The robot that visits $y$ cannot pass the point $y+\frac{I(y)}{2}<x$. Also the robot that visits $x$ cannot pass the point $x+\frac{I(x)}{2}$. Since $x\in S_{00}$ then $x+\frac{I(x)}{2}<1$. This means that no robot can visit point $y_n=1$.
\qed 
\end{proof}

\section{A Simple 4-Approximation Patrolling Schedule}
\label{sec: 4-feasible}
In light of Theorem~\ref{thm:no-00-point}, it is interesting to study feasible instances of \patrol\ that may have points that cannot be patrolled by one robot, i.e. for which $S_{00} \not = \emptyset$. 
As a warm-up, we provide a 4-feasible patrolling schedule for such instances.
The advantage of this schedule is that  robots' trajectories are simple and no coordination is required.
%
\begin{theorem}
\label{thm:4-approximation}
Feasible instances of \patrol\ admit an 4-approximate patrolling schedule.
\end{theorem}
\begin{proof}[Theorem \ref{thm:4-approximation}]
Let $A$ be a feasible solution. Let $I=\min_{y\in S}I(y)$ and let $x\in S$ be such that $I(x) = I$.
If $I\geq \frac 12$, then one robot patrolling the interval $[0,1]$ gives a 4-approximation solution. Thus, we may assume that $I\leq \frac 12$. 

According to Observation \ref{obs:nec1}, at least one robot stays in $R(x)$ during $A$, at each time $t$. We claim that a nested traversal $\mathcal{A}$ in which one robot traverses $[0,1]$ and the other robot traverses $R(x)$ is a 4-approximation.

We split the interval $[0,1]$ into $A = [0, a]$, $R(x) = [a, a+I] = [a, 1-b]$ and $B = [1-b, 1]$, where $a + I + b = 1$.
First, note that the waiting time of each $y\in R(x)$ during $\mathcal{A}$ is $w_{\mathcal{A}}(y) = 2I = 2I(x)\leq 2I(y)$.
Thus, it remains to show that $w_{\mathcal{A}}(y)\leq 4I(y)$ for each point $y\in A\cup B$.

Without loss of generality assume that $|A|=a< b=|B|$. Using the assumption $I\leq \frac 12$ and $a + I + b = 1$,
we have $a+b\geq \frac 12$, and therefore $b\geq\frac 14$.
Using Observation~\ref{obs:order}, we consider a feasible schedule $\mathcal{B}$ in which $r_1$ is always to the left of $r_2$.
By Observation \ref{obs:nec1}, at least one robot stays in $R(x)$ at each time during $\mathcal{B}$. We consider the following cases:
\begin{description}
\item[(Case $y\in A$):] 
As at each time moment there must exist a robot in $R(x)$, then in $\mathcal{B}$ robot $r_1$ has to stay in $R(x)$ while $r_2$ is traversing $B=[1-b,1]$ twice to visit $y_n=1$ and return to $R(x)$. Therefore the waiting time $w_{\mathcal{B}}$ satisfies
$I(y)\geq w_{\mathcal{B}}\geq 2b\geq 2 \frac 14=\frac 12.$
On the other hand
$w_{\mathcal{A}}(y)=2=4 \frac 12\leq 4 I(y).$
\item[(Case $y\in B$):]
Let $y'=y-(a+I)$, thus $y'$ is the distance of $y$ to $R(x)$.
Consider a time $t$ during the execution of $\mathcal{B}$ at which $r_1$ leaves $R(x)$ in order to visit the point 0. As $r_2$ must be in $R(x)$ at $t$, the last visit of $y$ before $t$ was at time $t'\leq t-y'$. Then, it has to stay in $R(x)$ for at least $2a + y'$. The time between two consecutive visits at $y$ is at least $t + 2a + y'-(t - y') = 2a + 2y'$. On the other hand, in order to visit 1, $r_2$ has also time at least $2(1-y')$ between two consecutive visits of y. Altogether
$w_{\mathcal{B}}(y)\geq \max\{2(a + y'), 2(b- y')\}.$
Thus
$w_{\mathcal{B}}(y)\geq \frac 12[2(a + y')+ 2(b- y')]=a+b\geq \frac 12.$
On the other hand $w_\mathcal{A}(y) = 2$ and thus $w_\mathcal{A}(y)\leq 4w_{\mathcal{B}}\leq 4I(y)$.
\end{description}
\qed
\end{proof}

\section{A $\sqrt{3}$-Approximation Patrolling Schedule}
\label{sec: sqrt3-feasible}

The bottleneck toward patrolling instances of \patrol\ is caused by points which require the coordination of both robots in order to be patrolled, i.e. instances in which $S_{00} \not= \emptyset$. In order to improve upon the 4-feasible schedule of Theorem \ref{thm:4-approximation}, we need to understand better the visitation requirements of points in $S_{00}$, as well as their relative positioning in the path to be patrolled. The result of our analysis, and our main contribution, is an elaborate $\sqrt{3}$-feasible patrolling schedule. 
\begin{theorem}
\label{thm:sqrt3-approximation}
Feasible instances of \patrol\ admit an efficient $\sqrt{3}$-approximate patrolling schedule.
\end{theorem}

In what follows, we explicitly assume that $S_{00} \not= \emptyset$, as otherwise, due to Theorem~\ref{thm:no-00-point}, we can easily find feasible schedules for instances of \patrol\ that admit feasible solutions. 
Next, we introduce a key notion to our algorithms. 
\begin{definition}
\label{def: expanding}
Given an instance of \patrol\ we identify critical points $x_1, \ldots,x_4$ that are defined as follows: 
$\bigcap_{x\in S_{00}} R(x)=[x_1, x_4]$, and $x_2,x_3$ are the leftmost and rightmost points  point in $S_{00}$, respectively. The instance is called $\alpha$-expanding if $x_1=\frac{\alpha}{1+\alpha}x_4$. 
\end{definition}

Theorem~\ref{thm:sqrt3-approximation} is an immediate corollary of the following Lemmata~\ref{lem: 1+2a-approx},~\ref{lem: (2+a)/(1+a)-approx}  that we prove in subsequent Sections~\ref{sec: 1+2a-approx},~\ref{sec: (2+a)/(1+a)-approx}, respectively. The lemmata are interesting in their own right, since they explicitly guarantee good approximate schedules as a function of the expansion of the given instance. 

\begin{lemma}
\label{lem: 1+2a-approx}
Feasible $\alpha$-expanding instances of \patrol\ admit an efficient $(1+2\alpha)$-approximate patrolling schedule.
\end{lemma}

\begin{lemma}
\label{lem: (2+a)/(1+a)-approx}
Feasible $\alpha$-expanding instances of \patrol\ admit an efficient $\frac{2+\alpha}{1+\alpha}$-approximate patrolling schedule.
\end{lemma}

Lemmata~\ref{lem: 1+2a-approx},~\ref{lem: (2+a)/(1+a)-approx} above imply that any feasible $\alpha$-expanding instance admits a 
$\min\left\{
1+2\alpha, \frac{2+\alpha}{1+\alpha}
\right\}$ feasible patrolling schedule. The achieved approximation is the worst when the instance is $\frac{\sqrt{3}-1}{2}$-expanding, in which case, the patrolling schedule is $\sqrt{3}$-feasible. This concludes the proof of Theorem~\ref{thm:sqrt3-approximation}.

Notably, our feasibility bounds above are tight. In 
Appendices~\ref{sec: tight Algo1} and~\ref{sec: tight Algo2} 
we show that for every $\alpha$, there are feasible $\alpha$-expanding \patrol\ instances for which the performance of our patrolling schedules that prove Lemma \ref{lem: 1+2a-approx} and Lemma~\ref{lem: (2+a)/(1+a)-approx} (see Sections~\ref{sec: 1+2a-approx},~\ref{sec: (2+a)/(1+a)-approx}) is equal to the proposed bound. Hence, the performance analysis of our patrolling schedule showing Theorem~\ref{thm:sqrt3-approximation} cannot be improved.

\subsection{Useful Observations for Feasible \patrol\ Instances }
\label{sec:observations}

In an $\alpha$-expanding instance of \patrol\ we have that $x_1=\alpha(x_4-x_1)$. 
If the instance is also feasible, then by Lemma \ref{lem:nec-condition} we have that $S_{00}\subset \bigcap_{x\in S} R(x)$. Since $S_{00}\subset S$, we obtain that $S_{00}\subset\bigcap_{x\in S_{00}} R(x)=[x_1, x_4]$. Also, it is easy to see that for the critical points $x_1, \ldots, x_4$ we have that $x_1\leq x_2 <x_4$ and that $x_1 < x_3 \leq x_4$. In particular we may assume, without loss of generality, that $x_1\leq 1-x_4$, as otherwise we flip the order of all points. 
\ignore{
Without loss of generality we can assume that $x_1\leq 1-x_4$. Then let
\begin{itemize}
\item $x_1=\alpha(x_4-x_1)$, where $\alpha\in\mathbb{R}$.
\item $x_2\in [x_1,x_4):$ the location of the leftmost point of $S_{00}$.
\item $x_3\in (x_1,x_4]:$ the location of the rightmost point of $S_{00}$.
\end{itemize}
}
Also using Observation~\ref{obs:order}, we assume that the feasible schedule to the \patrol\  instance has robot $r_1$ stay always to the left of $r_2$. 
\begin{lemma}
\label{lem:nec2}
Consider a feasible patrolling schedule $\mathcal{A}$ for a \patrol\ instance. Then 
\begin{enumerate}[(1)]
\item there is always a robot inside the interval $[x_1,x_4]$. 
\item the interval $[0,x_1)$ is only traversed by $r_1$ and the interval $(x_4,1]$ is only traversed by $r_2$. 
\item $0\in R(x)$ for all $x\in [0,x_1)$, and $1\in R(x)$ for all $x\in (x_4,1]$.
\item $x_4-x_3\leq x_3 -x_1$ and $x_2-x_1\leq x_4-x_1$.
\end{enumerate}
\end{lemma}
\begin{proof}
The proof of (1) is a direct consequence of Observation \ref{obs:nec1} and the fact that $[x_1,x_4]$ is the intersection of the ranges of all of the points of $S_{00}$. 

During the execution of $\mathcal{A}$ a robot needs to visit 0 and 1. Also, by (1) we know that there is always a robot inside $[x_1,x_4]$. Therefore while the robot $r_2$ is traversing $(x_4,1]$ the robot $r_1$ has to stay inside $[x_1,x_4]$, and while robot $r_1$ is traversing $[0,x_1)$, the robot $r_2$ has to stay inside $[x_1,x_4]$. This implies that $r_1$ never passes $x_4$ and $r_2$ never passes $x_1$. This proves (2). Part (3) follows directly from (2).

We now prove the first inequality of (4). Suppose to the contrary that $x_4-x_3>x_3-x_1$, and thus $x_3<\frac{x_1+x_4}{2}$. For all $x\in S_{00}$ we have that $x_4\in R(x)$. Therefore for all $x\in S_{00}$, $x_4\leq x+\frac{I(x)}{2}$. Moreover $x_3$ is the rightmost point of $S_{00}$, hence $x\leq x_3< \frac{x_1+x_4}{2}$. Consequently
$x_4\leq x+\frac{I(x)}{2}\leq x_3+\frac{I(x)}{2}<\frac{x_1+x_4}{2}+\frac{I(x)}{2}.$
This implies that $I(x)> \frac{x_4-x_1}{2}$. So for all $x\in S_{00}$ we have 
$x-\frac{I(x)}{2}\leq x-\frac{x_4-x_1}{2}<\frac{x_1+x_4}{2}-\frac{x_4-x_1}{2}=x_1.$
Therefore there is a point $y\in (0,1)$ such that for all $x\in S_{00}$, $x-\frac{I(x)}{2}\leq y< x_1$. Hence $y\in \bigcap_{x\in S_{00}} R(x)$ and $y< x_1$. This contradicts the fact that $x_1$ is the leftmost point of the intersection of the ranges of all the points of $S_{00}$.
The proof of the second inequality of (4) follows by an analogous argument.
\qed
\end{proof}
\begin{lemma}
\label{lem:idle-time}
If there is a feasible solution for patrolling with two robots then the idle time of the points of $S$ satisfy the following inequalities.
$$
I(x)\geq\left\{
\begin{array}{ll}
\max\{2x,2(1-x-x_4+x_1),x_4-x_1\} &, x\in [0,x_1) \\
2\max\{x_4-x, x-x_1\} &, x\in [x_1,x_4] \\
\max\{2(1-x),2(x-x_4+x_1),x_4-x_1\} &, x\in (x_4,1]\\
\end{array}
\right.
$$
\end{lemma}
\begin{proof}
Let $\mathcal{A}$ be a feasible solution and $x\in S$. 

First assume that $x\in [0,x_1)$. 
By (2) of Lemma~\ref{lem:nec2}, in $\mathcal{A}$ the points of $[0,x_1)$ are only visited by $r_1$ and $0\in R(x)$. Thus, $I(x)\geq w_{\mathcal{A}}(x)\geq 2x$. 
Moreover robot $r_1$ has to stay inside the interval $[x_1,x_4]$ for at least $2(1-x_4)$ while the robot $r_2$ is traversing the interval $(x_4,1]$ to visit 1. The time length for $r_1$ to traverse from $x$ to $x_1$, stay for at least $2(1-x_4)$ inside $[x_1,x_4]$, and then traverse from $x_1$ to $x$ is at least $2[(x_1-x)+(1-x_4)]$. Therefore, $I(x)\geq w_{\mathcal{A}(x)}\geq 2[(x_1-x)+(1-x_4)]$.
On the other hand, by Lemma \ref{lem:nec-condition}, we know that $x_3\in R(x)$, and thus $\frac{I(x)}{2}\geq x_3-x\geq x_3-x_1.$ By (3) of Lemma~\ref{lem:nec2}, $x_4-x_3\leq x_3-x_1\leq \frac{I(x)}{2}$. Therefore, 
$x_4-x_1= (x_4-x_3)+(x_3-x_1)\leq \frac{I(x)}{2}+\frac{I(x)}{2}= I(x).$
By the above discussion, and for all $x\in [0,x_1)$, we have
$I(x)\geq \max\{2x,2(1-x-x_4+x_1),x_4-x_1\}.$
A similar argument shows that for $x\in(x_4,1]$ we have that 
$I(x)\geq\max\{2x,2(1-x-x_4+x_1),x_4-x_1\}.$

Now assume that $x\in [x_1,x_4]$. Then $x_1, x_4\in R(x)$, and therefore $x-\frac{I(x)}{2}\leq x_1\leq x_4\leq x+\frac{I(x)}{2}$. This implies that $2(x-x_1)\leq I(x)$ and $2(x_4-x)\leq I(x)$. So for all $x\in [x_1,x_4]$ we have
$I(x)\geq\max\{x-x_1, x_4-x\}.$
\end{proof}
\ignore{
In sequel we present two algorithms for patrolling with two robots in the presence of points of type $S_{00}$. We will prove that the feasibility ratio of these algorithms are functions of $\alpha$. Recall that $\alpha$ is a positive real number for which $x_1=\alpha(x_4-x_1)$. 
}

\subsection{$(1+2\alpha)$-Approximate Patrolling Schedules (Proof of Lemma~\ref{lem: 1+2a-approx})}
\label{sec: 1+2a-approx}

Given a feasible $\alpha$-expanding instance of \patrol\, and using its critical points as in Definition~\ref{def: expanding}, we propose the following algorithm. 
\vspace{-0.5cm}
\begin{algorithm}[H]
\caption{}
\label{alg:partition-patrolling}
\begin{algorithmic}[1]
\State{Robot $r_1$ starts anywhere in $[0,x_3]$, and robot $r_2$ starts anywhere in $[x_3,1]$.}
\State{Repeat forever}
	\Indent
	\State{Robot $r_1$ zigzags inside $[0,x_3]$ and robot $r_2$ zigzags inside $[x_3,1]$.}
	\EndIndent
\end{algorithmic}
\end{algorithm}
\vspace{-0.7cm}
Next we show that Algorithm~\ref{alg:partition-patrolling} is $(1+2\alpha)$-feasible, effectively proving Lemma~\ref{lem: 1+2a-approx}. For this we analyze the waiting time $w(x)$ for all points $x\in S$. 

Assume that $x\in [0,x_1)$. By Lemma \ref{lem:nec-condition}, we know that $x_3\in R(x)$. Moreover by (3) of Lemma~\ref{lem:nec2}, $0\in R(x)$. Since $r_1$ zigzags inside $[0,x_3]$ then $w(x)\leq I(x)$. 

Similarly, for $x\in (x_3,1]$, by Lemma \ref{lem:nec-condition} and Lemma~\ref{lem:nec2} we have $\{x_3, 1\}\subset R(x)$. Since $r_2$ zigzags inside $[x_3,1]$ then $w(x)\leq I(x)$. 

Finally, let $x\in [x_1,x_3]$. First assume that $x<x_3$. Then in Algorithm \ref{alg:partition-patrolling} the point $x$ is only visited by $r_1$. Since $r_1$ zigzags inside $[0,x_3]$ we have that
$w(x)=2\max\{x, x_3-x\}.$
We now compute the feasibility ratio. Clearly for the points $x\in [0,x_1)\cup (x_3,1]$ we have that $\frac{w(x)}{I(x)}\leq 1$. So when $x\in [x_1,x_3]$, then by Lemma \ref{lem:idle-time}
$\frac{w(x)}{I(x)}\leq \frac{\max\{x, x_3-x\}}{\max\{x-x_1, x_4-x\}}.$
First let $\max\{x-x_1, x_4-x\}=x_4-x$. Then $x\leq \frac{x_1+x_4}{2}$. If $\max\{x, x_3-x\}=x_3-x$, as $x_3\leq x_4$ we have that $\frac{w(x)}{I(x)}\leq 1$. If , on the other hand, $\max\{x, x_3-x\}=x$, then we have 
\begin{align*}
\frac{w(x)}{I(x)}&\leq\frac{x}{x_4-x}\leq \frac{\frac{x_1+x_4}{2}}{x_4-\frac{x_1+x_4}{2}}\leq \frac{x_1+x_4}{x_4-x_1}\\
&=\frac{(x_4-x_1)+2x_1}{x_4-x_1}=1+\frac{2x_1}{x_4-x_1}\\
&=1+\frac{2x_1}{\frac{x_1}{\alpha}}   \hspace{.7cm}[\mbox{Using }x_1=\alpha(x_4-x_1)]\\
&= 1+2\alpha.
\end{align*}
Now let $\max\{x-x_1, x_4-x\}=x-x_1$. Then $x\geq \frac{x_1+x_4}{2}$. Moreover by (4) of Lemma~\ref{lem:nec2}, we have $x_3\geq \frac{x_1+x_4}{2}$. Therefore $x_3-x\leq x-x_1$. If $\max\{x,x_3-x\}=x_3-x$ then $\frac{w(x)}{I(x)}\leq 1$. So assume that $\max\{x,x_3-x\}=x$, in which case 
\begin{align*}
\frac{w(x)}{I(x)}&\leq\frac{x}{x-x_1}\leq \frac{x-x_1+x_1}{x-x_1}=1+\frac{x_1}{x-x_1}\\
&\leq 1+\frac{x_1}{\frac{x_1+x_4}{2}-x_1}=1+\frac{2x_1}{x_4-x_1}=1+\frac{2x_1}{\frac{x_1}{\alpha}}= 1+2\alpha.
\end{align*}
Therefore, Algorithm \ref{alg:partition-patrolling} is a $(1+2\alpha)$-approximation algorithm. Our analysis above is tight. For details, see Appendix~\ref{sec: tight Algo1}.

\subsection{$\frac{2+\alpha}{1+\alpha}$-Approximate Patrolling Schedules (Proof of Lemma~\ref{lem: (2+a)/(1+a)-approx})}
\label{sec: (2+a)/(1+a)-approx}

The distributed algorithm that achieves feasibility performance $\frac{2+\alpha}{1+\alpha}$ is quite elaborate. At a high level, the two robots maintain some distance that never drops below a certain carefully chosen threshold. During the execution of the patrolling schedule, there will always be some robot patrolling (zigzaging within) a certain subinterval defined by critical points of the given instance. When the robots move towards each other, and their distance reaches the certain threshold, then robots exchange roles; the previously zigzaging robot abandons the subinterval and goes to patrol its extreme point, while the other robot starts zigzaging within the subinterval. The formal description of our algorithm follows. The reader may also consult Figure \ref{fig:algo-d}. 
\vspace{-0.3cm}
\begin{figure}[!htb]
\begin{center}
\includegraphics[width=5cm]{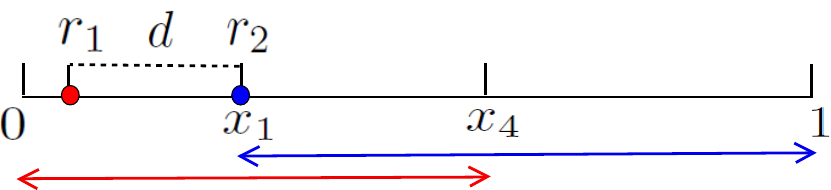}
\end{center}
\caption{The red arrow determines the patrolling area of $r_1$ and the blue arrow determines the patrolling area of $r_2$.}
\label{fig:algo-d}
\end{figure}

\vspace{-0.3cm}
\begin{algorithm}[H]
\caption{}
\label{alg:sqrt(3)}
\begin{algorithmic}[1]
\State{Let $d=\frac{1}{1+\alpha}\min\{x_1, x_4-x_1\}.$}
\State{Robot $r_1$ starts at $x_1-d$ and robot $r_2$ at $x_1$.}
\State{Repeat forever}
	\Indent
	\State{\textit{Patrolling Schedule of $r_1$:}}
	\While{$r_1$ is inside the interval $[x_1,x_4]$ and the distance between the locations of $r_1$ and $r_2$ is at least $d$}
		\State{Zigzag between points $x_1$ and $x_4$.}
	\EndWhile
	\State{Visit point 0, then visit point $x_1$, and then go to step 5.}
 	\EndIndent
 	
	\Indent
	\State{\textit{Patrolling Schedule of $r_2$:}}
	\While{$r_2$ is inside the interval $[x_1,x_4]$ and the distance between the locations of $r_2$ and $r_1$ is at least $d$}
		\State{Zigzag between points $x_1$ and $x_4$.}
	\EndWhile
	\State{Visit point 1, then visit point $x_4$, and then go to step 9.}	
 	\EndIndent
\end{algorithmic}
\end{algorithm}
\vspace{-0.7cm}
Note that each robot has an explicit segment in which the points are visited by only that robot, {\em i.e.} $[0,x_1)$ is the explicit segment of $r_1$ and $(x_4,1]$ is the explicit segment of $r_2$. The trajectories of the robots overlap at $[x_1,x_4]$ where the points are visited by both $r_1$ and $r_2$. The movements of the robots have two states: zigzagging inside $[x_1,x_4]$ and traversing their explicit segments twice.   
More precisely, once a robot enters $[x_1,x_4]$ it zigzags inside $[x_1,x_4]$ until the other robot is at distance $d$. Then it leaves $[x_1,x_4]$, traverses its explicit segment twice, and the same process repeats perpetually. 

Note that by the definition of $d$, we know that 
$\min\{x_1, x_4-x_1, 1-x_4\}\geq d.$
Therefore, the original placement of $r_1$ at $x_1-d$ is compatible with Algorithm~\ref{alg:sqrt(3)}.
The remaining of the section is devoted to proving that Algorithm~\ref{alg:sqrt(3)} is $\frac{2+\alpha}{1+\alpha}$-approximate, effectively proving Lemma~\ref{lem: (2+a)/(1+a)-approx}. As a first step, we calculate the worst case waiting times $w(x)$ of all points in $S$. 
\begin{lemma}
\label{lem:waiting-time}
The waiting times of points in $S$ for Algorithm \ref{alg:sqrt(3)} are as follows.
$$
w(x)
\left\{
\begin{array}{ll}
= 2\max\{x,1-x-d\} &, x\in [0,x_1) \\
\leq 2\max\{x-x_1, x_4-x\}+d &, x\in [x_1,x_4] \\
= 2\max\{1-x,x-d\} &, x\in (x_4,1]
\end{array}
\right.
$$
\end{lemma}

\begin{proof}
Recall that $x_1\leq 1-x_4$, and in particular $\min\{x_1, x_4-x_1, 1-x_4\}\geq d$. 

\begin{description}
\item[Case $0\leq x<x_1$:] Point $x$ is only visited by robot $r_1$. We now calculate the time interval between two consecutive visitations of $x$ by $r_1$. We distinguish two subcases. 

{\bf First} consider the subcase where $r_1$ is moving to the left when it visits $x$. Before $r_1$ visits $x$ again, it has to visit 0 and then return to $x$. Therefore, the time between the two visitations of $x$ is $2x$.  

{\bf Second} consider the subcase in which $r_1$ is moving to the right when it visits $x$. Before $r_1$ visits $x$ again, it has to visit $x_1$ (i.e. enter interval $[x_1, x_4]$), zigzag between points $x_1$ and $x_4$ until its distance to the other robots becomes $d$, and then $r_1$ exits the interval $[x_1,x_4]$ and return to $x$. Below we compute the total time between these two visitations of $x$ by $r_1$. 
\begin{itemize}
\item[(1a):] $r_1$ traverses from $x$ to $x_1$: it takes $x_1-x$. 
\item[(1b):] $r_1$ zigzags inside $[x_1,x_4]$: 
at the time that $r_1$ arrives at $x_1$ and starts  zigzaging inside $[x_1,x_4]$, robot $r_2$ is at distance $d$ from $r_1$ and it is moving to the right to visit 1 and return. Also, at the time that $r_1$ arrives at $x_1$ to exit the interval $[x_1,x_4]$, the distance between $r_1$ and $r_2$ is $d$ and robot $r_2$ is moving to the left to zigzag inside the interval $[x_1,x_4]$. Therefore, the time $r_1$ spends inside the interval $[x_1,x_4]$ is equal to the time that $r_2$ spends to traverse from $x_1+d$ to 1 and return to $x_1+d$, which is $2(1-x_1-d)$.
\item[(1c):] $r_1$ traverses from $x_1$ to $x$: it takes $x_1-x$.
\end{itemize}
Using (1a,1b,1c) above, we conclude that the total time between two consecutive visitations of $x$ by $r_1$ is $2(1-x-d)$.

Taking into consideration both subcases, the overall (worst case) waiting time of $x$ is $2\max\{x,1-x-d\}$.

\item[Case $x_4<x\leq 1$:] 
 The analysis is analogous to the previous case. 

\item[Case $x_1\leq x\leq x_4$:] 
Point $x$ is visited by both $r_1$ and $r_2$. We consider two subcases
\begin{itemize}
\item[(1)] The two consecutive visits of $x$ are by the same robot $r_1$ or $r_2$: this case occurs when either of $r_1$ or $r_2$ zigzags inside the interval $[x_1,x_4]$. Therefore $w(x)=2\max\{x_4-x, x-x_1\}$.
\item[(2)] The two consecutive visits of $x$ are by different robots $r_1$ and $r_2$: this case occurs when one robot is exiting the interval $[x_1,x_4]$ and the other one is entering it. 

First suppose that $r_1$ visits $x$ and the next visit of $x$ is performed by $r_2$. The worst waiting time in this case occurs when $r_1$ is about to visit $x$ but the distance between $r_1$ and $r_2$ reduces to $d$ and so $r_1$ turns away from $x$. Then $r_2$ visits $x$ after at most $d$ time steps. Note that since $x_1\geq d$ the visit of $x$ by $r_2$ is guaranteed. Therefore $w(x)\leq 2(x-x_1)+d$. Now assume that $r_2$ visits $x$ and the next visit of $x$ is performed by $r_1$. By a similar discussion we have that $w(x)\leq 2(x_4-x)+d$. This implies that $w(x)\leq 2\max\{x-x_1, x_4-x\}+d$.
\end{itemize}
By Subcases 1,2 above we conclude that $w(x)\leq 2\max\{x-x_1, x_4-x\}+d$, for all $x\in [x_1,x_4]$.
\end{description}
~
\qed
\end{proof}

The proof of Lemma~\ref{lem: (2+a)/(1+a)-approx} follows by upper bounding $\max_{x\in S}\left\{\frac{w(x)}{I(x)}\right\}$. Using Lemma~\ref{lem:idle-time} and Lemma~\ref{lem:waiting-time}, we see that the approximation ratio of Algorithm~\ref{alg:sqrt(3)} is no more than 
\begin{equation}
\label{equa: feas ratio}
\frac{w(x)}{I(x)}
\leq
\left\{
\begin{array}{ll}
\frac{2\max\{x,1-x-d\}}{\max\{2x,2(1-x-x_4+x_1),x_4-x_1\}} &, x\in [0,x_1) \\
\frac{2\max\{x-x_1, x_4-x\}+d}{2\max\{x_4-x, x-x_1\}} &, x\in [x_1,x_4] \\
\frac{2\max\{1-x,x-d\}}{\max\{2(1-x),2(x-x_4+x_1),x_4-x_1\}} &, x\in (x_4,1]
\end{array}
\right.
\end{equation}
Using that $d=\frac{\min\{x_1,x_4-x_1\}}{1+a}$, and the fact that the given instance is $\alpha$-expanding, 
i.e. that $x_1=\alpha(x_4-x_1)$, and after some tedious and purely algebraic calculations, we see that $\frac{w(x)}{I(x)} \leq \frac{2+\alpha}{1+\alpha}$ for all $x\in S$, as wanted. Details can be found in 
Appendix~\ref{sec: tedious cal for feasibility ratio}.

\section{Conclusion}

The paper investigated the problem of patrolling a line segment by two robots when time-patrolling constraints are placed on the frequency of visitation of all the points of the line. As shown in this study, finding ``efficient'' trajectories that satisfy the requirements or even deciding on their existence for two robots turns out to be a highly intricate problem.  Nothing better is known aside from the $\sqrt 3$-approximation algorithm for two robots on a line presented in this work. The patrolling problem with constraints is also open for more general graph topologies (e.g, cycles, trees, etc.). Further, the case of patrolling with constraints for multiple robots is completely unexplored in all topologies, including for the line segment.



\newpage
\appendix

\section{Tightness Analysis of Algorithm~\ref{alg:partition-patrolling}}
\label{sec: tight Algo1}

For every $\alpha$, we provide a feasible $\alpha$-expanding \patrol\ instance for which the performance of Algorithm~\ref{alg:partition-patrolling} is exactly equal to $1+2\alpha$. The instance is as follows. Choose $x_1\in (0,1)$ such that $2x_1+\frac{x_1}{\alpha}=1$. Consider three points $y_0=0$, $y_1=x_1(1+\frac{1}{2\alpha})$, and $y_2=1$. Let $I(y_0)=I(y_2)=4x_1+\frac{x_1}{\alpha}$ and $I(y_1)=\frac{x_1}{\alpha}$. 

Let $\mathcal{A}$ be a solution as follows. At the beginning, the robot $r_1$ locates at $x_1$ and the robot $r_2$ locates at $y_1$, and both robots move to the right. See Figure \ref{fig:partition-tight}.
\begin{figure}[!htb]
\begin{center}
\includegraphics[width=7cm]{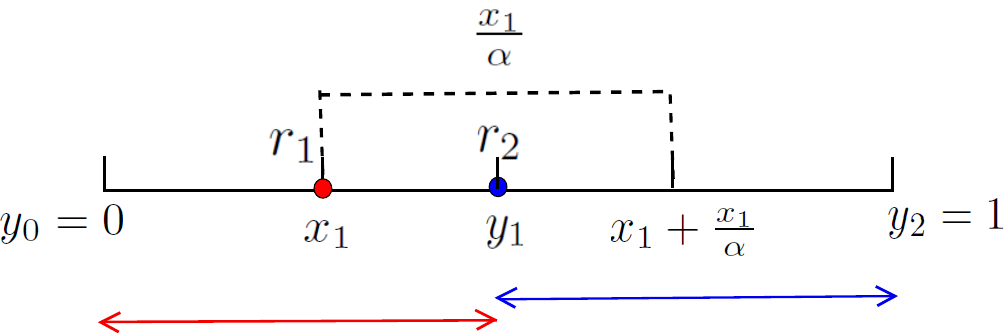}
\end{center}
\caption{The red arrow determines the patrolling area of $r_1$ in $\mathcal{A}$ and the blue arrow determines the patrolling area of $r_2$ in $\mathcal{A}$.}
\label{fig:partition-tight}
\end{figure}
The patrolling segment of $r_1$ is $[y_0,y_1]$ and the patrolling segment of $r_2$ is $[y_1,y_2]$. The robot $r_1$ moves back and forth inside the interval $[y_0,y_1]$ and each time $r_1$ visits $y_1$ stays at $y_1$ for $2x_1$. Similarly the robot $r_2$ moves back and forth inside the interval $[y_1,y_2]$ and each time it visits $y_1$ stays at $y_1$ for $2x_1$. Then $w_{\mathcal{A}}(y_0)=w_{\mathcal{A}}(y_2)=4x_1+\frac{x_1}{\alpha}=I(y_0)=I(y_1)$ and $w_{\mathcal{A}}(y_1)= \frac{x_1}{2\alpha}<I(y_1).$ Therefore $\mathcal{A}$ is a feasible solution. 

Now consider Algorithm \ref{alg:partition-patrolling}. For the above example we have $x_2=x_3=y_1$. Therefore $r_1$ zigzags inside $[0,y_1]$ and $r_2$ zigzags inside $[y_1,1]$. Since the movement of the robots is not coordinated in Algorithm \ref{alg:partition-patrolling} the worst waiting time of $y_1$ in Algorithm \ref{alg:partition-patrolling} is $2x_1+\frac{x_1}{\alpha}$. This implies that 
$$\frac{w(x)}{I(x)}=\frac{x_1(2+\frac{1}{\alpha})}{\frac{x_1}{\alpha}}=1+2\alpha.$$

\section{Tightness Analysis of Algorithm~\ref{alg:sqrt(3)}}
\label{sec: tight Algo2}

For every $\alpha$, we provide a feasible $\alpha$-expanding \patrol\ instance for which the performance of Algorithm~\ref{alg:sqrt(3)} is exactly equal to $\frac{2+\alpha}{1+\alpha}$. For this we consider two cases.

\vspace{.3cm}
\noindent{\bf Case 1. $\alpha< 1$:} Choose $x_1\in (0,1)$ such that $x_1+2\frac{x_1}{\alpha}=1$. Consider four points $y_0=0$, $y_1=x_1+\frac{x_1}{2\alpha}$, $y_2=\frac{2x_1}{\alpha}$, and $y_3=1$. See Figure \ref{fig:algo-d-tight}. Let $I(y_0)=I(y_3)=x_1(2+\frac{3}{\alpha})$, $I(y_1)=\frac{x_1}{\alpha}$, and $I(y_2)=\frac{2x_1}{\alpha}$.

\begin{figure}[!htb]
\begin{center}
\includegraphics[width=8cm]{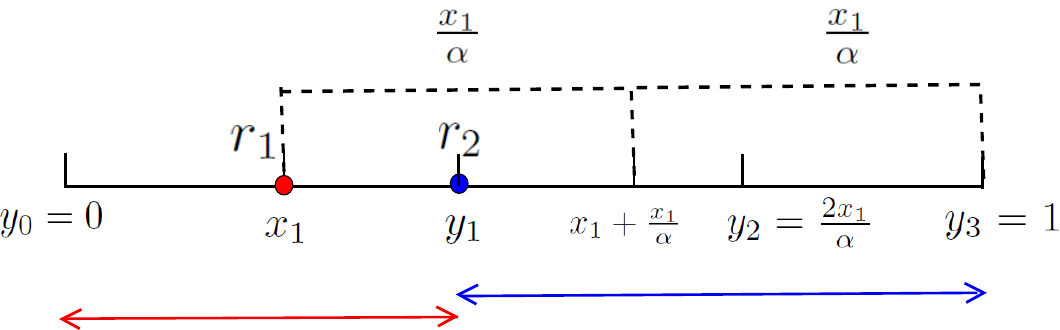}
\end{center}
\caption{The red arrow determines the patrolling area of $r_1$ in $\mathcal{A}$ and the blue arrow determines the patrolling area of $r_2$ in $\mathcal{A}$.}
\label{fig:algo-d-tight}
\end{figure}

Let $\mathcal{A}$ be a solution as follows. 
\begin{enumerate}[(1)]
\item At the beginning, the robot $r_1$ locates at $x_1$ and the robot $r_2$ locates at $y_1$, and both robots move to the right. 
\item The patrolling segment of $r_1$ is $[y_0,y_1]$ and the patrolling segment of $r_2$ is $[y_1,y_3]$. 
\item The robot $r_1$ moves back and forth inside the interval $[y_0,y_1]$ and each time $r_1$ visits $y_1$ stays at $y_1$ until $r_2$ is at point $y_1+\frac{x_1}{\alpha}$ and is moving to the left. Then $r_1$ moves to the left, visits 0 and returns to $y_1$.
\item The robot $r_2$ moves back and forth inside the interval $[y_1,y_3]$ and each time it visits $y_1$ stays at $y_1$ for $2x_1$. Then $r_2$ moves to the right, visits $y_2$ and 1 and returns to $y_1$.
\end{enumerate}
First we analyze the movement of the robots in $\mathcal{A}$. The robot $r_1$ leaves $y_1$ when the distance between $r_1$ and $r_2$ is $\frac{x_1}{\alpha}$. Note that this is possible since the length of $[y_1,y_3]$ is greater than $\frac{x_1}{\alpha}$. The next visit of $r_1$ from $y_1$ occurs after
$$\frac{x_1}{2\alpha}+2x_1+\frac{x_1}{2\alpha}=\frac{x_1}{\alpha}+2x_1,$$
which is the time that $r_1$ spends to visit 0 and return to $y_1$. During the time $\frac{x_1}{\alpha}+2x_1$ the robots $r_1$ traverse from $y_1+\frac{x_1}{\alpha}$ to $y_1$ and stays there for $2x_1$. Therefore by the time $r_1$ arrives at $y_1$, the robot $r_2$ leaves $y_1$. We now compute the waiting time of $\mathcal{A}$.

The point $y_0$ is only visited by $r_1$. So the waiting time of 0 is equal \lq\lq two times the length of the interval $[0,y_1]$ plus the time $r_1$ stops at $y_1$''. The stop time of $r_1$ at $y_1$ is equal the time $r_2$ traverse from $y_1$ to 1 and return to the point $y_1+\frac{x_1}{\alpha}$ which is $\frac{2x_1}{\alpha}$. Therefore 
$$w_{\mathcal{A}}(y_0)=\frac{2x_1}{\alpha}+2x_1+\frac{x_1}{\alpha}=x_1(2+\frac{3}{\alpha})=I(x).$$
The point $y_1$ is visited by both $r_1, r_2$. First suppose that $r_1$ is waiting at point $y_1$. The robot $r_1$ leaves $y_1$ when  $r_2$ is at distance $\frac{x_1}{\alpha}$ of $y_1$. So the next visit of $y_1$ occurs after $\frac{x_1}{\alpha}$ time by $r_2$. Also, as we discussed above by the time $r_2$ is ready to leave $y_1$ the robot $r_1$ arrives at $y_1$. So in this case there is no time interval between the visits of $r_2$ and $r_1$. Therefore
$$w_{\mathcal{A}}(y_1)=\frac{x_1}{\alpha}=I(x).$$
The point $y_2$ is only visited by $r_2$. If $r_2$ is moving to the left when it visits $y_2$ then the next visit of $y_2$ occurs in $\frac{2x_1}{\alpha}$. This is the time $r_2$ spends to visit $y_1$ and stays at it for $2x_1$ and returns to $y_2$. If $r_2$ is moving to the right when it visits $y_2$ then clearly the next visit of $y_2$ occurs in $2(1-\frac{2x_1}{\alpha})=2x_1$. Therefore
$$w_{\mathcal{A}}(y_2)=\max\{2x_1, 2\frac{x_1}{\alpha}\}=\frac{2x_1}{\alpha}=I(x).$$
For $y_3=1$ it is easy to see that 
$$w_{\mathcal{A}}(y_3)=2(1-y_1)+2x_1=2(1-x_1-\frac{x_1}{2\alpha})+2x_1=2-\frac{x_1}{\alpha}=x_1(2+\frac{3}{\alpha}).$$
The last equation follows from the fact that $x_1+\frac{2x_1}{\alpha}=1.$ Therefore $\mathcal{A}$ is a feasible solution.

Now consider Algorithm \ref{alg:sqrt(3)}. For the above example we have $d=\frac{x_1}{1+\alpha}$, $x_2=x_3=y_1$, and $x_4=x_1+\frac{x_1}{\alpha}<\frac{2x_1}{\alpha}$. Therefore by Lemma \ref{lem:waiting-time} 
$$w(y_2)= 2\max\{x_1, \frac{2x_1}{\alpha}-\frac{x_1}{1+\alpha}\}=\frac{2x_1}{\alpha}-\frac{x_1}{1+\alpha}.$$
Therefore 
$$\frac{w(y_2)}{I(y_2)}=\frac{\frac{2x_1}{\alpha}-\frac{x_1}{1+\alpha}}{\frac{x_1}{\alpha}}=1+\frac{x_1(\frac{1}{1+\alpha})}{\frac{x_1}{\alpha}}=\frac{2+\alpha}{1+\alpha}.$$

\vspace{.3cm}
\noindent{\bf Case 2. $\alpha\geq 1$:} Choose $x_1\in (0,1)$ such that $2x_1+\frac{x_1}{\alpha}=1$ and let $0<\epsilon<x_1$. Consider four points $y_0=0$, $y_1=x_1-\epsilon$, $y_2=x_1+\frac{x_1}{2\alpha}$, and $y_3=1$. See Figure \ref{fig:algo-d-tight-2}. Let $I(y_0)=I(y_3)=2(1-x_1)$, $I(y_1)=2(x_1+\epsilon)$, and $I(y_2)=\frac{x_1}{\alpha}$.

\begin{figure}[!htb]
\begin{center}
\includegraphics[width=8cm]{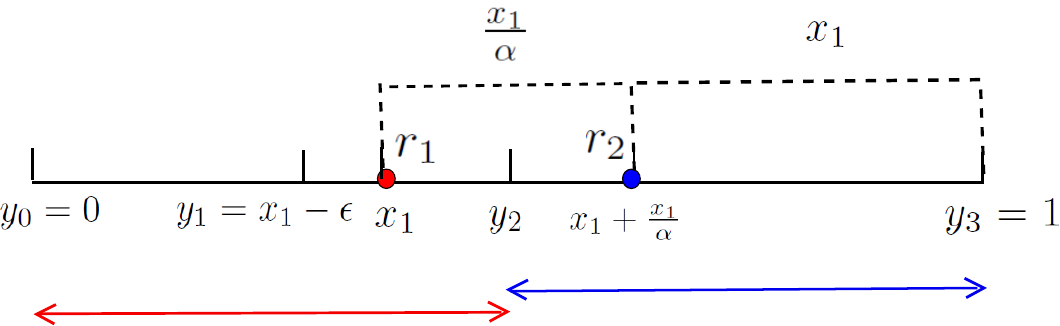}
\end{center}
\caption{The red arrow determines the patrolling area of $r_1$ in $\mathcal{A}$ and the blue arrow determines the patrolling area of $r_2$ in $\mathcal{A}$.}
\label{fig:algo-d-tight-2}
\end{figure}
Let $\mathcal{A}$ be a solution as follows. 
\begin{enumerate}[(1)]
\item At the beginning, the robot $r_1$ locates at $x_1$ and the robot $r_2$ locates at $x_1+\frac{x_1}{\alpha}$, and both robots move to the right. 
\item The patrolling segment of $r_1$ is $[y_0,y_2]$ and the patrolling segment of $r_2$ is $[y_2,y_3]$. 
\item The robot $r_1$ moves back and forth inside the interval $[y_0,y_2]$ and each time $r_1$ visits $y_1$ stays at $y_1$ for $\frac{x_1}{\alpha}$. Then $r_1$ moves to the left, visits $y_1$ and 0 and returns to $y_1$.
\item The robot $r_2$ moves back and forth inside the interval $[y_2,y_3]$ and each time it visits $y_1$ stays at $y_1$ for $\frac{x_1}{\alpha}$. Then $r_2$ moves to the right, visits 1 and returns to $y_1$.
\end{enumerate}
First we analyze the movement of the robots in $\mathcal{A}$. The robot $r_1$ leaves $y_1$ when the distance between $r_1$ and $r_2$ is $\frac{x_1}{\alpha}$, and similarly the robot $r_2$ leaves $y_1$ when the distance between $r_1$ and $r_2$ is $\frac{x_1}{\alpha}$. This is possible since the length of the intervals $[0,y_2]$ and $[y_2,1]$ is $x_1+\frac{x_1}{2\alpha}$ and $x_1\geq \frac{x_1}{\alpha}$. 
It is easy to see that the 
$$w_{\mathcal{A}}(y_0)=w_{\mathcal{A}}(y_3)=2(1-x_1).$$
The point $y_1=x_1-\epsilon$ is only visited by $r_1$ and so
\begin{align*}
w(y_1)&=2\max\{y_1, y_2-y_1+\frac{x_1}{2\alpha}\}\\
&=2\max\{x_1-\epsilon, \frac{x_1}{\alpha}+\epsilon\}\leq 2(x_1+\epsilon).
\end{align*}
In $\mathcal{A}$ the point $y_2=x_1+\frac{x_1}{2\alpha}$ is patrolled by both $r_1$ and $r_2$ in such a way that one robot stays at $y_2$ until the other robot is moving towards $y_2$ and is at distance $\frac{x_1}{\alpha}$ from $y_2$. This implies that 
$w_{\mathcal{A}}(y_2)\leq\frac{x_1}{\alpha}=I(x)$.

From the above discussion we have that $\mathcal{A}$ is a feasible solution.
Now consider Algorithm \ref{alg:sqrt(3)}. For the above example we have $d=\frac{x_1}{\alpha(1+\alpha)}$, $x_2=x_3=y_2$. Since $y_1=x_1-\epsilon<x_1$ by Lemma \ref{lem:waiting-time} we have that 
\begin{align*}
w(y_1)&= 2\max\{x_1-\epsilon, 1-(x_1-\epsilon)-\frac{x_1}{\alpha(1+\alpha)}\}\\
&=2\max\{x_1-\epsilon, x_1+\epsilon+x_1(\frac{1}{1+\alpha})\}=x_1+\epsilon+x_1(\frac{1}{1+\alpha}).
\end{align*}
In the above equation we use the facts that $2x_1+\frac{x_1}{\alpha}=1$ and $\alpha\geq 1$.
Therefore 
\begin{align*}
\frac{w(y_1)}{I(y_1)}&=\frac{x_1+\epsilon+x_1(\frac{1}{1+\alpha})}{x_1+\epsilon}\\
&=1+\frac{1}{1+\alpha}\frac{x_1}{x_1+\epsilon}.
\end{align*}
This implies that as $\epsilon$ converges to 0 the feasibility ratio of Algorithm \ref{alg:sqrt(3)} converges to $1+\frac{1}{1+\alpha}=\frac{2+\alpha}{1+\alpha}$.

\section{Omitted Proofs of Section~\ref{sec: (2+a)/(1+a)-approx} }
\label{sec: tedious cal for feasibility ratio}

Using that $d=\frac{\min\{x_1,x_4-x_1\}}{1+a}$, and the fact that the given instance is $\alpha$-expanding, 
i.e. that $x_1=\alpha(x_4-x_1)$, we prove that $w(x)/I(x)$, as it reads in~\eqref{equa: feas ratio}, is upper bounded by $\frac{2+\alpha}{1+\alpha}$. Naturally, we distinguish cases for the location of $x$ with respect to the three subintervals $[1,x_1), [x_1,x_4], (x_4, 1-x_4]$. For the sake of exposition, we split the proof in 3 corresponding Lemmata.

\begin{lemma}
For all $0\leq x <x_1$ we have $\frac{w(x)}{I(x)}\leq \frac{2+\alpha}{1+\alpha}$.
\end{lemma}

\begin{proof}
Using~\eqref{equa: feas ratio} we have that
\begin{equation*}
\frac{w(x)}{I(x)}\leq \frac{\max\{x,1-x-d\}}{\max\{x,1-x-x_4+x_1,x_4-x_1\}}.
\end{equation*}
If $x>\frac{1-d}{2}$ then $\max\{x,1-x-d\}=x$, and so $\frac{w(x)}{I(x)}\leq 1$. So assume that $x\leq\frac{1-d}{2}$. We consider two cases:

\vspace{.3cm}
\noindent {\bf Case 1: $\alpha\leq 1$.} Then $x_1=\alpha(x_4-x_1)\leq x_4-x_1$, and so $d=\frac{1}{1+\alpha} x_1$.
\begin{equation*}
\frac{w(x)}{I(x)}\leq \frac{1-x-\frac{x_1}{1+\alpha}}{\max\{x,1-x-\frac{x_1}{\alpha},\frac{x_1}{\alpha}\}}\leq \frac{1-x-\frac{x_1}{1+\alpha}}{\max\{1-x-\frac{x_1}{\alpha},\frac{x_1}{\alpha}\}}.
\end{equation*} 
First let $\max\{1-x-\frac{x_1}{\alpha},\frac{x_1}{\alpha}\}=\frac{x_1}{\alpha}$. Then $\frac{x_1}{\alpha}\geq 1-x-\frac{x_1}{\alpha}$, which implies that $x\geq 1-\frac{2x_1}{\alpha}$. Therefore
\begin{align*}
\frac{w(x)}{I(x)} 
&\leq \frac{1-x-\frac{x_1}{1+\alpha}}{\frac{x_1}{\alpha}}\leq \frac{1-(1-\frac{2x_1}{\alpha})-\frac{x_1}{1+\alpha}}{\frac{x_1}{\alpha}}\\
&= \frac{\frac{2x_1}{\alpha}-\frac{x_1}{1+\alpha}}{\frac{x_1}{\alpha}}= 1+\frac{x_1(\frac{1}{\alpha}-\frac{1}{1+\alpha})}{\frac{x_1}{\alpha}}\\
&= 1+\frac{1}{1+\alpha}=\frac{2+\alpha}{1+\alpha}.
\end{align*}
Now let $\max\{1-x-\frac{x_1}{\alpha},\frac{x_1}{\alpha}\}=1-x-\frac{x_1}{\alpha}$. Then $\frac{x_1}{\alpha}\leq 1-x-\frac{x_1}{\alpha}$, which implies that $x\leq 1-\frac{2x_1}{\alpha}$. 
\begin{align*}
\frac{w(x)}{I(x)} 
&\leq \frac{1-x-\frac{x_1}{1+\alpha}}{1-x-\frac{x_1}{\alpha}}= \frac{1-x-\frac{x_1}{\alpha}+\frac{x_1}{\alpha}-\frac{x_1}{1+\alpha}}{1-x-\frac{x_1}{\alpha}}\\
&= 1+\frac{\frac{x_1}{\alpha(1+\alpha)}}{1-x-\frac{x_1}{\alpha}}\leq 1+ \frac{\frac{x_1}{\alpha(1+\alpha)}}{1-(1-\frac{2x_1}{\alpha})-\frac{x_1}{\alpha}}\\
&= 1+\frac{1}{1+\alpha}=\frac{2+\alpha}{1+\alpha}.
\end{align*}

\noindent{\bf Case 2: $\alpha\geq 1$.} Then $x_1=\alpha(x_4-x_1)\geq x_4-x_1$, and $d=\frac{1}{1+\alpha} (x_4-x_1)=\frac{1}{\alpha(1+\alpha)}x_1$.
\begin{equation*}
\frac{w(x)}{I(x)}\leq \frac{1-x-\frac{x_1}{\alpha(1+\alpha)}}{\max\{x,1-x-\frac{x_1}{\alpha},\frac{x_1}{\alpha}\}}\leq \frac{1-x-\frac{x_1}{\alpha(1+\alpha)}}{\max\{x, 1-x-\frac{x_1}{\alpha}\}}.
\end{equation*} 
First let $\max\{x, 1-x-\frac{x_1}{\alpha}\}=x$. Then $x\geq 1-x-\frac{x_1}{\alpha}$, which implies that $x\geq \frac 12-\frac{x_1}{2\alpha}$. 
\begin{align*}
\frac{w(x)}{I(x)} 
&\leq \frac{1-x-\frac{ x_1}{\alpha(1+\alpha)}}{x}\leq \frac{1-(\frac 12-\frac{x_1}{2\alpha})-\frac{ x_1}{\alpha(1+\alpha)}}{\frac 12-\frac{x_1}{2\alpha}}\\
&= \frac{\frac 12-\frac{x_1}{2\alpha}+(\frac{x_1}{\alpha}-\frac{x_1}{\alpha (1+\alpha)})}{\frac 12-\frac{x_1}{2\alpha}}= 1+\frac{x_1(\frac{1}{1+\alpha})}{\frac 12-\frac{x_1}{2\alpha}}.
\end{align*}
By assumption we know that $x_1\leq 1-x_4$, and $x_4-x_1=\frac{x_1}{\alpha}$. Moreover $x_1+(x_4-x_1)+(1-x_4)=1$. Therefore $x_1+\frac{x_1}{\alpha}+x_1\leq 1$, which implies that $x_1\leq \frac{\alpha}{2\alpha+1}$. Therefore 
\begin{align*}
1+\frac{x_1(\frac{1}{1+\alpha})}{\frac 12-\frac{x_1}{2\alpha}}&\leq 1+\frac{\frac{\alpha}{2\alpha+1}(\frac{1}{1+\alpha})}{\frac 12-\frac{1}{2(2\alpha+1)}}=1+\frac{\frac{\alpha}{2\alpha+1}(\frac{1}{1+\alpha})}{\frac{\alpha}{2\alpha+1}}\\
& = 1+\frac{1}{1+\alpha}=\frac{2+\alpha}{1+\alpha}. 
\end{align*}
Now let $\max\{x,1-x-\frac{x_1}{\alpha}\}=1-x-\frac{x_1}{\alpha}$. Then $x\leq 1-x-\frac{x_1}{\alpha}$, which implies that $x\geq 1-\frac{2x_1}{\alpha}$. In the following calculations we use the fact that $x\leq x_1$.
\begin{align*}
\frac{w(x)}{I(x)} 
&\leq\frac{1-x-\frac{ x_1}{\alpha(1+\alpha)}}{1-x-\frac{x_1}{\alpha}}= \frac{1-x-\frac{x_1}{\alpha}+\frac{x_1}{\alpha}-\frac{x_1}{\alpha (1+\alpha)}}{1-x-\frac{x_1}{\alpha}}\\
&= 1+\frac{x_1(\frac{1}{\alpha}-\frac{1}{\alpha(1+\alpha)})}{1-x-\frac{x_1}{\alpha}}\leq 1+\frac{x_1(\frac{1}{\alpha}-\frac{1}{\alpha(1+\alpha)})}{1-x_1-\frac{x_1}{\alpha}}\\
&\leq 1+\frac{\frac{\alpha}{2\alpha+1}(\frac{1}{1+\alpha})}{1-\frac{\alpha}{2\alpha+1}(1+\frac{1}{\alpha})}
= 1+\frac{\frac{\alpha}{2\alpha+1}(\frac{1}{1+\alpha})}{1-\frac{1+\alpha}{2\alpha+1}}\\
&= 1+\frac{\frac{\alpha}{2\alpha+1}(\frac{1}{1+\alpha})}{\frac{\alpha}{2\alpha+1}}
=1+\frac{1}{1+\alpha}=\frac{2+\alpha}{1+\alpha}.
\end{align*}
This proves that for all $x\in [0,x_1)$ the ratio $\frac{w(x)}{I(x)}$ is at most $\frac{2+\alpha}{1+\alpha}$.
\qed \end{proof}

\begin{lemma}
For all $x_1\leq x \leq x_4$ we have $\frac{w(x)}{I(x)}\leq \frac{2+\alpha}{1+\alpha}$.
\end{lemma}

\begin{proof}
Using~\eqref{equa: feas ratio} we have that
\begin{equation*}
\frac{w(x)}{I(x)}\leq \frac{\max\{x-x_1, x_4-x\}+\frac d2}{\max\{x-x_1, x_4-x\}}\\
= 1+ \frac{d}{\max\{x-x_1, x_4-x\}}.
\end{equation*}
If $x-x_1\leq x_4-x$ then $x\leq \frac{x_1+ x_4}{2}$ and $\max\{x-x_1, x_4-x\}=x_4-x$. Moreover $x_4-x\geq x_4-\frac{x_1+x_4}{2}=x_4-x_1$. Therefore
$$\frac{w(x)}{I(x)}\leq 1+\frac{d}{x_4-x_1}.$$
Since $d\leq \frac{1}{1+\alpha}(x_4-x_1)$ we have
$$\frac{w(x)}{I(x)}\leq 1+\frac{1}{1+\alpha}=\frac{2+\alpha}{1+\alpha}.$$
Now let $x_4-x\leq x- x_1$. Then $x\geq\frac{x_1+ x_4}{2}$ and $\max\{x-x_1, x_4-x\}=x-x_1$. Moreover, $x-x_1\geq \frac{x_1+x_4}{2}-x_1=x_4-x_1$, and thus
$$\frac{w(x)}{I(x)}\leq 1+\frac{d}{x_4-x_1}\leq 1+\frac{1}{1+\alpha}=\frac{2+\alpha}{1+\alpha}.$$
\qed \end{proof}

\begin{lemma}
For all $x_4< x \leq 1$ we have $\frac{w(x)}{I(x)}\leq \frac{2+\alpha}{1+\alpha}$.
\end{lemma}

\begin{proof}
Using~\eqref{equa: feas ratio} we have that
\begin{equation*}
\frac{w(x)}{I(x)}\leq \frac{\max\{1-x,x-d\}}{\max\{1-x,x-x_4+x_1,x_4-x_1\}}.
\end{equation*}
If $x>\frac{1-d}{2}$ then $\max\{1-x,x-d\}=1-x$, and so $\frac{w(x)}{I(x)}\leq 1$. So assume that $x\leq\frac{1-d}{2}$. We consider two cases:

\vspace{.3cm}
\noindent {\bf Case 1: $\alpha\leq 1$.} Then $x_1=\alpha(x_4-x_1)\leq x_4-x_1$, and so $d=\frac{1}{1+\alpha} x_1$.
\begin{equation*}
\frac{w(x)}{I(x)}\leq \frac{x-\frac{x_1}{1+\alpha}}{\max\{1-x,x-\frac{x_1}{\alpha},\frac{x_1}{\alpha}\}}\leq \frac{x-\frac{x_1}{1+\alpha}}{\max\{x-\frac{x_1}{\alpha},\frac{x_1}{\alpha}\}}.
\end{equation*} 
First let $\max\{x-\frac{x_1}{\alpha},\frac{x_1}{\alpha}\}=\frac{x_1}{\alpha}$. Then $\frac{x_1}{\alpha}\geq x-\frac{x_1}{\alpha}$, which implies that $x\leq\frac{2x_1}{\alpha}$. Therefore
\begin{align*}
\frac{w(x)}{I(x)} 
&\leq \frac{x-\frac{x_1}{1+\alpha}}{\frac{x_1}{\alpha}}\leq \frac{\frac{2x_1}{\alpha}-\frac{x_1}{1+\alpha}}{\frac{x_1}{\alpha}}\\
&= \frac{\frac{2x_1}{\alpha}-\frac{x_1}{1+\alpha}}{\frac{x_1}{\alpha}}= 1+\frac{x_1(\frac{1}{\alpha}-\frac{1}{1+\alpha})}{\frac{x_1}{\alpha}}\\
&= 1+\frac{1}{1+\alpha}=\frac{2+\alpha}{1+\alpha}.
\end{align*}
Now let $\max\{x-\frac{x_1}{\alpha},\frac{x_1}{\alpha}\}=x-\frac{x_1}{\alpha}$. Then $\frac{x_1}{\alpha}\leq x-\frac{x_1}{\alpha}$, which implies that $x\geq \frac{2x_1}{\alpha}$. 
\begin{align*}
\frac{w(x)}{I(x)} 
&\leq \frac{x-\frac{x_1}{1+\alpha}}{x-\frac{x_1}{\alpha}}= \frac{x-\frac{x_1}{\alpha}+\frac{x_1}{\alpha}-\frac{x_1}{1+\alpha}}{x-\frac{x_1}{\alpha}}\\
&= 1+\frac{\frac{x_1}{\alpha(1+\alpha)}}{x-\frac{x_1}{\alpha}}\leq 1+ \frac{\frac{x_1}{\alpha(1+\alpha)}}{\frac{2x_1}{\alpha}-\frac{x_1}{\alpha}}\\
&= 1+\frac{1}{1+\alpha}=\frac{2+\alpha}{1+\alpha}.
\end{align*}

\noindent{\bf Case 2: $\alpha\geq 1$.} Then $x_1=\alpha(x_4-x_1)\geq x_4-x_1$, and $d=\frac{1}{1+\alpha} (x_4-x_1)=\frac{1}{\alpha(1+\alpha)}x_1$.
\begin{equation*}
\frac{w(x)}{I(x)}\leq \frac{x-\frac{x_1}{\alpha(1+\alpha)}}{\max\{1-x,x-\frac{x_1}{\alpha},\frac{x_1}{\alpha}\}}\leq \frac{x-\frac{x_1}{\alpha(1+\alpha)}}{\max\{1-x, x-\frac{x_1}{\alpha}\}}.
\end{equation*} 
First let $\max\{1-x, x-\frac{x_1}{\alpha}\}=1-x$. Then $1-x\geq x-\frac{x_1}{\alpha}$, which implies that $x\leq \frac 12+\frac{x_1}{2\alpha}$. 
\begin{align*}
\frac{w(x)}{I(x)} 
&\leq \frac{x-\frac{x_1}{\alpha(1+\alpha)}}{1-x}\leq \frac{\frac 12+\frac{x_1}{2\alpha}-\frac{ x_1}{\alpha(1+\alpha)}}{\frac 12-\frac{x_1}{2\alpha}}\\
&= \frac{\frac 12-\frac{x_1}{2\alpha}+(\frac{x_1}{\alpha}-\frac{x_1}{\alpha (1+\alpha)})}{\frac 12-\frac{x_1}{2\alpha}}= 1+\frac{x_1(\frac{1}{1+\alpha})}{\frac 12-\frac{x_1}{2\alpha}}.
\end{align*}
By assumption we know that $x_1\leq 1-x_4$, and $x_4-x_1=\frac{x_1}{\alpha}$. Moreover $x_1+(x_4-x_1)+(1-x_4)=1$. Therefore $x_1+\frac{x_1}{\alpha}+x_1\leq 1$, which implies that $x_1\leq \frac{\alpha}{2\alpha+1}$. Therefore 
\begin{align*}
1+\frac{x_1(\frac{1}{1+\alpha})}{\frac 12-\frac{x_1}{2\alpha}}&\leq 1+\frac{\frac{\alpha}{2\alpha+1}(\frac{1}{1+\alpha})}{\frac 12-\frac{1}{2(2\alpha+1)}}=1+\frac{\frac{\alpha}{2\alpha+1}(\frac{1}{1+\alpha})}{\frac{\alpha}{2\alpha+1}}\\
& = 1+\frac{1}{1+\alpha}=\frac{2+\alpha}{1+\alpha}. 
\end{align*}
Now let $\max\{1-x,x-\frac{x_1}{\alpha}\}=x-\frac{x_1}{\alpha}$. In the calculations we use the fact that $x\geq x_4$ and $x_4=\frac{1+\alpha}{\alpha}x_1$ (recall that $x_1=\alpha (x_4-x_1)$).
\begin{align*}
\frac{w(x)}{I(x)} 
&\leq\frac{x-\frac{ x_1}{\alpha(1+\alpha)}}{x-\frac{x_1}{\alpha}}= \frac{x-\frac{x_1}{\alpha}+\frac{x_1}{\alpha}-\frac{x_1}{\alpha (1+\alpha)}}{x-\frac{x_1}{\alpha}}\\
&= 1+\frac{x_1(\frac{1}{\alpha}-\frac{1}{\alpha(1+\alpha)})}{x-\frac{x_1}{\alpha}}\leq 1+\frac{x_1(\frac{1}{1+\alpha})}{x_4-\frac{x_1}{\alpha}}\\
&= 1+\frac{x_1(\frac{1}{1+\alpha})}{\frac{1+\alpha}{\alpha}x_1-\frac{1}{\alpha}x_1}
=1+\frac{1}{1+\alpha}=\frac{2+\alpha}{1+\alpha}.
\end{align*}
This proves that for all $x\in (x_4,1]$ the ratio $\frac{w(x)}{I(x)}$ is at most $\frac{2+\alpha}{1+\alpha}$. 
\qed \end{proof}

\end{document}